\documentclass[letterpaper, 10 pt, conference]{ieeeconf}  
\IEEEoverridecommandlockouts                              
\usepackage{mathtools}
\usepackage{bm}
\usepackage{bbm}

\usepackage{graphicx}
\usepackage{amssymb}
\usepackage{color}

\newcommand{\R}{\mathbb{R}} 

\newcommand{\CA}{\mathcal{A}} %
\newcommand{\CC}{\mathcal{C}} %
\newcommand{\CE}{\mathcal{E}} %

\newcommand{\col}{{\rm col}}

\usepackage{amsthm}

\theoremstyle{definition}
\newtheorem{definition}{Definition}
\newtheorem{remark}{Remark}

\newtheorem{lemma}{Lemma}

\theoremstyle{plain}
\newtheorem{theorem}{Theorem}

\newcommand*{\QE}{\hfill\ensuremath{\square}}%
\newenvironment{pfof}[1]{\vspace{1ex}\noindent{\textit{Proof of
		#1:}}\hspace{0.5em}} {\hfill\QE\vspace{1ex}}

\usepackage{xcolor}
\definecolor{myblue}{HTML}{e0f3f8}
\definecolor{mygreen}{HTML}{abdda4}
\definecolor{mylightpurple}{HTML}{fde0ef}
\definecolor{darkgreen}{HTML}{4d9221}

\usepackage{cite}
\usepackage{amsmath}
\usepackage{tikz}
\usepackage{mathtools}

\usepackage[colorlinks]{hyperref}
\hypersetup{citecolor=[RGB]{5,112,176},linkcolor=[RGB]{5,112,176},urlcolor=[RGB]{5,112,176}}

\title{\LARGE \bf Control of Seizure Dynamics in a Bistable System with Pulsatile Inputs}

\author{Xianlong Zhao and Yuzhen Qin
\thanks{This work was supported in part by the project Dutch Brain Interface Initiative (DBI$^2$) with project number 024.005.022 of the research programme Gravitation which is (partly) financed by the Dutch Research Council (NWO).
X. Zhao (remote research intern) and Y. Qin is with the Department of Machine Learning and Neural Computing, Donders Institute for Brain, Cognition and Behaviour, Radboud University, Nijmegen, the Netherlands. (Contact: yuzhen.qin@donders.ru.nl).  
}
}

\begin{document}

\maketitle
\thispagestyle{empty}
\pagestyle{empty}

\begin{abstract}
    Bistable neural population models describe transitions associated with epilepsy through the coexistence of a stable equilibrium representing normal neural activity and a stable limit cycle representing seizure-like oscillations. Perturbations that drive the system outside the equilibrium’s basin of attraction can trigger transitions to seizure-like activity. Enlarging this basin enables the system to recover normal activity from a larger set of perturbed states, providing a model-based mechanism for reducing frequency of seizure initiation. Motivated by electrical brain stimulation, this paper investigates basin enlargement using pulsatile feedback control. Combining higher-order averaging with periodic near-identity transformations, we establish conditions under which nonzero-mean and zero-mean pulsatile feedback gains enlarge the basin of attraction and derive guidelines for pulse design. Numerical simulations validate the theoretical enlargement.
\end{abstract}


\section{Introduction}
Brain stimulation offers a treatment option for patients with drug-resistant epilepsy who are unsuitable for resective surgery.  Although stimulation can reduce seizure frequency, responses vary and many patients continue to experience seizures~\cite{nair2020nine}. Stimulation parameters are often selected empirically, and therapeutic mechanisms remain incompletely understood. These limitations motivate theoretical approaches that connect stimulation waveforms to neural dynamics and guide predictable control design. 

We investigate whether stimulation-like inputs can enlarge the healthy state's basin of attraction in a bistable neural population model. Its stable equilibrium and stable limit cycle represent normal activity and seizure-like oscillations, respectively~\cite{2010_KSN_VDN_epilepsy_bistable}. Perturbations can initiate seizure-like activity by driving the state outside the equilibrium's basin. Local stability alone therefore does not quantify resilience to such transitions. Enlarging the basin enables recovery from a larger set of perturbed states, including states that would otherwise approach the seizure-like attractor, providing a model-based objective for seizure frequency reduction.

\textbf{Related work.}
Fractional-order approaches support model-based seizure regulation through network stabilization~\cite{wang2026stabilizing,2023_REA_RG_BP_PS_ACC} and predictive neurostimulation~\cite{chatterjee2020fractional}. Passivity-based feedback has also been used to establish local stabilization and estimate the basin of attraction of the Epileptor neural mass model~\cite{acharya2026passivity}.
Seizure regulation is closely related to controlling pathological oscillations~\cite{2022_WD_MJ}, including optimal inputs for terminating seizure-like bursting~\cite{wilson2014hamilton} and network interventions for containing their spread~\cite{allibhoy2022optimal}.
Bistable models capture the coexistence of healthy and seizure-like attractors~\cite{2010_KSN_VDN_epilepsy_bistable}. Our earlier work analytically characterizes the healthy-state basin of attraction and relates its size to seizure susceptibility~\cite{self_qin2024analytical}. Pinning-like feedback has further been developed to stabilize the non-ictal equilibrium in networks of bistable oscillators~\cite{liu2025synchronization}.

\begin{figure}[t]
\centering
\includegraphics[width=0.32\textwidth]{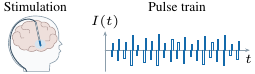}
\vspace{-8pt}
\caption{Pulsatile inputs in electrical brain stimulation }
\label{Pulsatile}
\vspace{-15pt}
\end{figure}

Practical brain-stimulation systems (see Fig.~\ref{Pulsatile}) deliver trains of electrical pulses~\cite{krauss2021technology}, commonly using charge-balanced waveforms to avoid net direct current and limit irreversible electrochemical reactions at the electrode-tissue interface and brain tissue damage~\cite{merrill2005electrical}. These waveform properties motivate studying how pulsatile control affects seizure susceptibility. Previous work studies stabilization of oscillator networks by high-frequency, zero-mean control~\cite{2023_own_OJCS}, but limited to local stability analysis. We investigate how pulsatile feedback reshapes the basin boundary.

\textbf{Contribution.}
This paper investigates enlargement of the healthy equilibrium's basin of attraction using pulsatile feedback that has non-zero and zero means. Unlike most existing studies that rely on local stability analysis, we combine first- and higher-order averaging with periodic near-identity transformations to approximate the basins of attraction of the controlled systems. We analytically show that the basins of attraction under control strictly contain the uncontrolled one, indicating that enlargement is established. The obtained results also provide a way to design pulse parameters to maximize the  enlargement under practical constraints. Numerical simulations confirm the theoretical enlargement.

\section{Problem Formulation}
We employ the following system to describe the activity of  neuronal populations:
\begin{subequations}\label{main}
    \begin{align}
    \dot x &= - \omega y + x f(x, y) \,, \\
    \dot y &= \hphantom{-}\omega x + y f(x, y) \,,
\end{align}
\end{subequations}
where 
\begin{equation}\label{func_f}
    f(x,y) \coloneq \sigma + 2 a b (x^2+y^2)-b(x^2+y^2)^2,
\end{equation}
and $x,y\in \R$ describe the population-averaged firing rate of pyramidal and interneuron populations, respectively \cite{2014_KS_KM_PG_et_al,2010_KSN_VDN_epilepsy_bistable}. The parameters satisfy
$\omega, a, b>0$ and $\sigma\in \R$. 

This system exhibits various dynamical behaviors with different parameter configurations \cite{2018_Srogatz_nonlinear_dyna,2012_BO_FTHB_Math_epilepsy}.  In particular, it is  \textit{bistable} if $-a^2b < \sigma < 0$, where a locally stable  equilibrium $\CE\coloneq\{\col(x,y)=0\}$ and a locally stable limit cycle 
$
	\CC \coloneq\left\{x,y\in\R\colon x^2 +y^2 = {a+ \sqrt{a^2 + \sigma/b}}\right\}
$
coexist, with basins of attraction separated by the unstable limit cycle:
\begin{equation*}
	\CC' \coloneq\left\{x,y\in\R\colon x^2 +y^2 = {a- \sqrt{a^2 + \sigma/b}}\right\}.
\end{equation*}
In other words, the system converges to $\CE$ if it starts inside $\CC'$ and converges to $\CC$ if it starts outside of $\CC'$. 

In this bistable regime, the system captures the key dynamical features associated with epileptic seizures. The stable equilibrium represents normal baseline neural activity, while the limit cycle corresponds to seizure-like oscillations. Seizure onset and termination are modeled as transitions between these two states, driven by either intrinsic fluctuations or external perturbations, consistent with the spontaneous  nature of epileptic activity. 
As demonstrated in~\cite{self_qin2024analytical}, the basin of attraction defined by $\CC'$ is crucial for distinguishing healthy from epileptic brain dynamics. A smaller basin of attraction implies greater sensitivity to perturbations and, consequently, a higher susceptibility to seizure initiation, a characteristic commonly associated with epilepsy.

Throughout this paper, we assume that the model parameters satisfy the conditions required for bistability, and the term basin of attraction refers to the basin of attraction of the equilibrium $\CE$ defined by
\begin{equation*}
	\CA^{\rm unctr} \coloneq\left\{x,y\in\R\colon x^2 +y^2 <\rho_{\mathrm{unctr}}:= {a- \sqrt{a^2 + \sigma/b}} \right\},
\end{equation*}

\textbf{Objective:} We aim to design feedback control to enlarge the basin of attraction, thereby increasing the system’s resilience to fluctuations and perturbations and reducing its susceptibility to seizure onset.

Particularly, we consider the controlled system described
\begin{subequations}\label{control}
	\begin{align}
		\dot x &= - \omega y + x f(x, y) +u_x(t) \,, \\
		\dot y &= \hphantom{-}\omega x + y f(x, y) +u_y(t)\,,
	\end{align}
\end{subequations}
where the feedback control inputs are linear, given by 
\begin{subequations}
    \begin{align}
        &u_x(t) = -K(t)x\\
        &u_y(t) = -K(t)y.
    \end{align}
\end{subequations} 
\begin{definition}[Basins of attraction]
For a prescribed feedback gain $K(t)$, let $z(t;t_0,z_0)$ denote the solution of \eqref{control} initialized at $z(t_0)=z_0$, where $z=\operatorname{col}(x,y)$. The basin of attraction of the origin at initial time $t_0$ is defined as
\begin{equation*}
\mathcal A_K(t_0)
:=
\left\{
z_0\in\mathbb R^2:
\lim_{t\to\infty}\|z(t;t_0,z_0)\|=0
\right\}.
\end{equation*}
Note that the considered feedback control can make the system time varying, thus the basin of attraction may depend on the initial time $t_0$.
\end{definition}

For the simple linear feedback where  $K(t)=k>0$, one can derive that the basin of attraction is enlarged to  
\begin{equation*}
	\CA^{\rm lin} \coloneq\left\{x,y\in\R\colon x^2 +y^2 < a-\sqrt{a^2+(\sigma-k)/{b}} \right\}.
\end{equation*}
However, this type of continuous control is generally undesirable in electrical neurostimulation due to safety constraints. This motivates us to consider control inputs that have stimulation waveforms in this paper.

To facilitate the analysis of the basin of attraction, we rewrite the controlled system~\eqref{control} in polar coordinates. By introducing \(x+iy=\varrho e^{i\theta}\) and defining the radial variable \(\rho=\varrho^2=x^2+y^2\), the system becomes
    \begin{align}
        \dot{\rho} &= 2\rho\bigl(h(\rho)-K(t)\bigr),\label{radial dynamics}\\
        \dot{\theta} &= \omega,
    \end{align}
where \(h(\rho)=\sigma+2ab\rho-b\rho^2\). Since the radial dynamics are decoupled from the angular dynamics, the evolution of \(\rho\) can be studied independently. Therefore, we focus on the radial dynamics in \eqref{radial dynamics}
to characterize how control inputs of stimulation form shape the basin of attraction of the system.

\section{Pulsatile feedback with non-zero mean}

Consider the pulse-like feedback gain
\begin{equation}\label{negative-mean-pulse}
K(t)=
\begin{cases}
k, & ((t-t_0)\bmod T)\in[0,dT),\\
0, & \text{otherwise},
\end{cases}
\end{equation}
where $t\ge t_0$, and $k>0$, $T>0$, and $d\in(0,1)$ denote the pulse amplitude, stimulation period, and duty cycle, respectively. Stimulation is initiated at the initial time $t_0$.  The theorem below provides a way to design the parameters $k, T$, and $d$ to enlarge the basin of attraction.

\begin{theorem}[Estimated basin under pulsatile control]\label{thm:pulse-convergence}
Suppose that the pulse-like feedback~\eqref{negative-mean-pulse} satisfies
\begin{equation}
    kd<a^2b+\sigma,
\end{equation}
Then, for any fixed $\kappa_0\in(0,kd(1-d)\rho_-^{\mathrm{avg}})$, there exists $T_0>0$ such that, for any $T \in (0,T_0 ]$,  every solution to the system~\eqref{control} initialized at time $t_0$ in the set 
\begin{equation*}
    \mathcal A_p
:=\left\{x,y\in\mathbb R:
x^2+y^2<\rho_-(\kappa_0,T)\right\}
\end{equation*}
remains in $\mathcal A_p$ for all $t\geq t_0$ and converges asymptotically to $\CE$. In particular, $\mathcal A_p\subseteq\mathcal A_K(t_0)$. Here, 
$\rho_-(\kappa_0,T) =\rho_-^{\mathrm{avg}} +(kd(1-d)\rho_-^{\mathrm{avg}}-\kappa_0)T$ with 
\begin{equation}\label{pulse-averaged-boundary}
\rho_-^{\mathrm{avg}}
=a-\sqrt{a^2+\frac{\sigma-kd}{b}}.
\end{equation}
\end{theorem}

Since $0<\kappa_0<kd(1-d)\rho_-^{\mathrm{avg}}$, consequently, for sufficiently small $T\in (0,T_0]$,
$$
\rho_-(\kappa_0,T)
=\rho_-^{\mathrm{avg}}
+\bigl(kd(1-d)\rho_-^{\mathrm{avg}}-\kappa_0\bigr)T
>\rho_-^{\mathrm{avg}}
>\rho_{\mathrm{unctr}}.
$$
Thus, the estimated basin $\mathcal A_p$ strictly contains the uncontrolled basin of attraction, i.e., $\CA^{\mathrm{unctr}}\subsetneq\mathcal A_p\subseteq\mathcal A_K(t_0)$, showing that the basin of attraction is enlarged.

\begin{remark}
    The set $\mathcal A_p$ is constructed from an asymptotic approximation of the basin boundary, characterized below through an unstable periodic solution of the radial dynamics. The analysis uses the unstable equilibrium of the averaged system to approximate this boundary. Thus, beyond establishing local stability of the origin, the theorem provides an explicit estimate of its basin of attraction, i.e., every trajectory initialized in $\mathcal A_p$  at $t_0$ remains in $\mathcal A_p$ and converges asymptotically to the origin.
\end{remark}

\subsection{Proof of Theorem~\ref{thm:pulse-convergence}}

First, we apply averaging techniques to the time-varying system~\eqref{radial dynamics} under the control input~\eqref{negative-mean-pulse}.
Let $s=(t-t_0)/T$ and $V(s):=K(Ts+t_0)$. Since $\int_0^1V(s)\,ds=kd$, then first-order averaging yields averaged radial system
\begin{equation}\label{avg pulse-like}
\frac{d\bar\rho}{ds}
=2T\bar\rho\bigl(h(\bar\rho)-kd\bigr).
\end{equation}

The following lemma determines when this averaged system remains bistable and characterizes the boundary of the basin of attraction of its origin.

\begin{lemma}[Bistability of averaged system]\label{lem:bistability-pulse}
Suppose that $-a^2b<\sigma<0$, $k>0$, and $d\in(0,1)$. The averaged system~\eqref{avg pulse-like} is bistable if and only if
\begin{equation}\label{pulse-bistability-condition}
kd<a^2b+\sigma.
\end{equation}
Under this condition,  the basin of attraction of $\bar\rho=0$ for~\eqref{avg pulse-like} is $[0,\rho_-^{\mathrm{avg}})$, where
\begin{equation}\label{pulse-averaged-boundary}
\rho_-^{\mathrm{avg}}
=a-\sqrt{a^2+\frac{\sigma-kd}{b}}.
\end{equation}
We call $\rho_-^{\mathrm{avg}}$ the \emph{averaged basin boundary}.
\end{lemma}

\begin{proof}
Note that $h(\bar\rho)-kd=\sigma+a^2b-kd-b(\bar\rho-a)^2$. If~\eqref{pulse-bistability-condition} holds, its zeros $\rho_\pm^{\mathrm{avg}}=a\pm\sqrt{a^2+(\sigma-kd)/b}$ are real, distinct, and positive since $\sigma-kd<0$. Moreover, $h(\bar\rho)-kd<0$ on $(0,\rho_-^{\mathrm{avg}})\cup(\rho_+^{\mathrm{avg}},\infty)$ and $h(\bar\rho)-kd>0$ on $(\rho_-^{\mathrm{avg}},\rho_+^{\mathrm{avg}})$. Since solutions of the scalar system~\eqref{avg pulse-like} are monotone, those starting in $[0,\rho_-^{\mathrm{avg}})$ converge to $0$, while those starting above $\rho_-^{\mathrm{avg}}$ converge to $\rho_+^{\mathrm{avg}}$. Hence, the system is bistable and the basin of $\bar\rho=0$ is $[0,\rho_-^{\mathrm{avg}})$. If~\eqref{pulse-bistability-condition} fails, $h(\bar\rho)-kd\leq0$ for all $\bar\rho\geq0$, so there is no stable positive equilibrium and the system is not bistable.
\end{proof}

Under the bistability condition,
$\rho_-^{\mathrm{avg}}$ can be used to 
approximate an unstable $T$-periodic solution of the
original radial dynamics. In the proof of Theorem~\ref{thm:pulse-convergence},
we use this periodic solution to identify initial conditions
whose trajectories converge to the origin.

\begin{lemma}[Existence of periodic solution]\label{lem:pulse-periodic-existence}
Suppose that condition~\eqref{pulse-bistability-condition} holds. Then, there exists $T_e>0$, such that for all $T \in (0,T_e]$, the radial system~\eqref{radial dynamics} under the pulse-like feedback~\eqref{negative-mean-pulse} admits a $T$-periodic solution $\rho_-^{\rm per}(t;T)$ satisfying
\begin{equation*}
\rho_-^{\rm per}(t_0;T)\to\rho_-^{\rm avg}\text{ when }T\to0.
\end{equation*}
\end{lemma}

\begin{proof}
Let $g_1(\rho) := 2\rho(h(\rho)-kd)$, under the bistability condition~\eqref{pulse-bistability-condition}, the averaged system~\eqref{avg pulse-like} has its equilibrium satisfying $g_1(\rho_-^{\rm avg}) = 0$, $g_1'(\rho_-^{avg}) \neq0$, and averaged function $g_1(\rho)$ is not identically zero. By \cite[Theorem A]{Novaes_2022}, there exists a $T$-periodic solution $\rho^{\rm per}_-(t;T)$ of radial system~\eqref{radial dynamics}, such that the stated limit above holds.
\end{proof}

We next approximate $\rho_-^{\rm per}(t;T)$ using the periodic near-identity transformation~\cite[Sec. 2.8]{Sanders2007Averaging}:
\begin{equation}\label{pulse-near-identity}
\rho=\bar\rho+Tu_1(\bar\rho,s)+\mathcal O(T^2),
\end{equation}
associated with first-order averaging,
where $s=(t-t_0)/T$ and $u_1$ satisfies 
\begin{equation}\label{pulse-homological-equation}
\frac{\partial u_1}{\partial s}
=2\bar\rho\bigl(kd-V(s)\bigr),
\qquad
\int_0^1u_1(\bar\rho,s)\,ds=0.
\end{equation}
Solving~\eqref{pulse-homological-equation} and evaluating $u_1(\bar \rho,s)$
at $\bar\rho=\rho_-^{\mathrm{avg}}$ and $s=0$, we arrive at
\begin{equation*}
u_1\bigl(\rho_-^{\mathrm{avg}},0\bigr)
=kd(1-d)\rho_-^{\mathrm{avg}}.
\end{equation*}
Since $s=0$ corresponds to the initial time $t=t_0$,
the unstable periodic solution satisfies
\begin{equation}\label{pulse-periodic-boundary}
\rho_-^{\mathrm{per}}(t_0;T)
=\rho_-^{\mathrm{avg}}
+kd(1-d)\rho_-^{\mathrm{avg}}T
+\mathcal O(T^2).
\end{equation}

\begin{lemma}[Upper bound of periodic solution]\label{lem:pulse-periodic-max}
Let $\rho_-^{\rm per}(t;T)$ be the unstable $T$-periodic solution of
system~\eqref{radial dynamics} under the pulse input~\eqref{negative-mean-pulse}.
For sufficiently small $T$, it holds that
\begin{equation*}
\max_{t\in[t_0,t_0+T]}\rho_-^{\rm per}(t;T)
=
\rho_-^{\rm per}(t_0;T).
\end{equation*}
\end{lemma}

\begin{proof}
Let $\rho_*(t):=\rho_-^{\rm per}(t;T)$. Since
$\rho_*(t)=\rho_-^{\rm avg}+\mathcal O(T)$ uniformly on $[t_0,t_0+T]$ and
$h(\rho_-^{\rm avg})=kd$, we have
$h(\rho_*(t))=kd+\mathcal O(T)$.
Hence, for sufficiently small $T$,
\begin{equation*}
\dot\rho_*(t)<0,\qquad t\in(t_0,t_0+dT),
\end{equation*}
because $K(t)=k$ and
$h(\rho_*(t))-k=-k(1-d)+\mathcal O(T)<0$, while
\begin{equation*}
\dot\rho_*(t)>0,\qquad t\in(t_0+dT,t_0+T),
\end{equation*}
because $K(t)=0$ and
$h(\rho_*(t))=kd+\mathcal O(T)>0$.

Thus $\rho_*(t)$ decreases on $(t_0,t_0+dT)$ and increases on $(t_0+dT,t_0+T)$. Since
$\rho_*(t_0+T)=\rho_*(t_0)$, it follows that
$\rho_*(t)\leq\rho_*(t_0)$ for all $t\in[t_0,t_0+T]$.
\end{proof}

\begin{pfof}{Theorem~\ref{thm:pulse-convergence}}
Consider the initial value problem
\begin{equation*}
\dot{\rho}=2\rho\bigl(h(\rho)-K(t)\bigr),
\qquad \rho(t_0)=\rho_0,
\end{equation*}
where $(x(t_0),y(t_0))\in\mathcal A_p$ so that initial value satisfies $\rho_0=x(t_0)^2+y(t_0)^2<\rho_-(\kappa_0,T)$.
Choose $0<\kappa_0<kd(1-d)\rho_-^{\rm avg}$.
By~\eqref{pulse-periodic-boundary},
$\rho_-^{\rm per}(t_0;T)-\rho_-(\kappa_0,T)
=\kappa_0T+\mathcal O(T^2)>0$
for sufficiently small $T>0$. Therefore,
\begin{equation*}
0\leq\rho(t_0)<\rho_-(\kappa_0,T)
<\rho_-^{\rm per}(t_0;T).
\end{equation*}
We only consider the nontrivial case $\rho(t_0)>0$,
since if $\rho(t_0)=0$, then $\rho(t)=0$
for all $t\geq t_0$.

Let $\rho_*(t):=\rho_-^{\rm per}(t;T)$.
By order preservation of the radial dynamics,
$0<\rho(t)<\rho_*(t)$ for all $t\geq t_0$.
Under condition~\eqref{pulse-bistability-condition},
$0<\rho_-^{\rm avg}<a$.
The periodic-solution expansion associated with
the transformation~\eqref{pulse-near-identity}
gives $\rho_*(t)=\rho_-^{\rm avg}+\mathcal O(T)$
uniformly over one period.
Thus, by periodicity, $\rho_*(t)<a$ for all
$t\geq t_0$ when $T$ is sufficiently small.


Define $\phi(t):=\log\bigl(\rho(t)/\rho_*(t)\bigr)$.
For all $t\geq t_0$,
\begin{equation*}
\dot\phi(t)
=2b\bigl(\rho(t)-\rho_*(t)\bigr)
     \bigl(2a-\rho(t)-\rho_*(t)\bigr)<0.
\end{equation*}
Hence,
$\rho(t+T)/\rho_*(t+T)<\rho(t)/\rho_*(t)$.
Since $\rho_*(t+T)=\rho_*(t)>0$, we obtain
$\rho(t+T)<\rho(t)$.

Let $\Pi_T$ denote the one-period map and define $\rho_n:=\rho(t_0+nT)$ for $n=0,1,\ldots$. Then
\begin{equation*}
0<\rho_{n+1}=\Pi_T(\rho_n)<\rho_n,
\end{equation*}
so $\{\rho_n\}$ converges to some $\rho_\infty\geq0$. By continuity of $\Pi_T$, taking limits on both sides of $\rho_{n+1}=\Pi_T(\rho_n)$ gives
\begin{equation*}
\rho_\infty=\Pi_T(\rho_\infty).
\end{equation*}
Suppose that $\rho_\infty>0$. Since $0<\rho_\infty\leq\rho_0<\rho_*(t_0)$, the one-period decrease established above yields $\Pi_T(\rho_\infty)<\rho_\infty$, contradicting this equality. Therefore, $\rho_\infty=0$.

Furthermore, since $h(\rho)\leq\sigma+a^2b$
and $K(t)\geq0$, we have
$\dot\rho/\rho\leq C:=2(\sigma+a^2b)$.
For $t\in[t_0+nT,t_0+(n+1)T]$, integration gives
$\rho(t)\leq e^{CT}\rho_n$.
Consequently, $\rho(t)\to0$, and hence
$(x(t),y(t))\to(0,0)$.

Finally, $\phi(t)\leq\phi(t_0)$ implies
$\rho(t)\leq\rho(t_0)\rho_*(t)/\rho_*(t_0)$.
By Lemma~\ref{lem:pulse-periodic-max} and periodicity,
$\rho_*(t)\leq\rho_*(t_0)$ for all $t\geq t_0$.
Therefore,
\begin{equation*}
\rho(t)\leq\rho(t_0)<\rho_-(\kappa_0,T),
\qquad t\geq t_0.
\end{equation*}
Choosing a common $T_0>0$ to satisfy all the
smallness conditions above, we conclude that,
for every $T\in(0,T_0]$, every solution initialized
in $\mathcal A_p$ at stimulation onset remains
in $\mathcal A_p$ for all $t\geq t_0$ and converges
asymptotically to the origin.\end{pfof}

\section{Pulsatile feedback with zero mean}
Motivated by electrical stimulation paradigms, where charge-balanced pulses are commonly used to  limit irreversible electrochemical reactions and brain tissue damage~\cite{merrill2005electrical}, we next consider zero-mean pulse-like feedback, with the feedback gain $K(t)$ described by
\begin{align}\label{pulse_zero_mean}
K(t)=
\begin{cases}
k_1, & ((t-t_0)\bmod T)\in[0,d_1T),\\
k_2, & ((t-t_0)\bmod T)\in[d_1T,(d_1+d_2)T),\\
0, & \text{otherwise},
\end{cases}
\end{align}
where $d_1,d_2>0$, $D:=d_1+d_2\leq1$, and
\begin{equation}\label{zero-mean-condition}
k_1d_1+k_2d_2=0.
\end{equation}

\begin{remark}
Condition~\eqref{zero-mean-condition} makes the gain $K(t)$ zero mean, but the state-dependent inputs $u_x(t)=-K(t)x,u_y(t)=-K(t)y$ may have non-zero mean over each cycle. However, when $T$ is small, $u_x(t), u_y(t)$ approaches zero mean for each cycle as $x$ and $y$ change much slower. We leave exact charge balance control for future studies.
\end{remark}

\begin{theorem}[Estimated basin under zero-mean feedback]\label{thm:zero-mean-convergence}
Assume $k_1>0$ and $k_2<0$. For any fixed $\kappa'>0$, there exists $T_1>0$ such that, for any $T\in(0,T_1]$, every solution to the system~\eqref{control} initialized at time $t_0$ in the set
\begin{equation*}
\mathcal A_z := \left\{x,y\in\mathbb R: x^2+y^2<\rho_-(\kappa',T)\right\}
\end{equation*}
remains in $\mathcal A_z$ for all $t\geq t_0$ and converges asymptotically to the origin. In particular, $\mathcal A_z\subseteq\mathcal A_K(t_0)$. Here,
\begin{equation}\label{basin_bound}
    \rho_-(\kappa',T) = \rho_{\rm unctr} + k_1d_1D\rho_{\rm unctr}T + (C_\rho-\kappa')T^2,
\end{equation}
with
\begin{equation*}
C_\rho = \Delta_{\rm avg} + \Gamma_0 + \frac12(k_1d_1D)^2\rho_{\rm unctr},
\end{equation*}
where
\begin{align*}
    &\Delta_{\rm avg} = \frac{D(4-3D)}{6\lambda}(a-\lambda)(a-2\lambda)(k_1d_1)^2,\\
    &\Gamma_0 = -\frac23 bk_1d_1D(4d_1+2d_2-3)\rho_{\rm unctr}^{\,2}(a-\rho_{\rm unctr}), \text{ and}\\
    &\lambda := \sqrt{a^2+\frac{\sigma}{b}}, \qquad \rho_{\rm unctr} := a-\lambda, \qquad D := d_1+d_2.
\end{align*}
\end{theorem}

Since $k_1d_1D\rho_{\rm unctr}>0$, the positive first-order term in \eqref{basin_bound} dominates the second-order correction for sufficiently small $T$. Hence, there exists $T_1$, for every $T\in(0,T_1]$,
\begin{equation*}
\rho_-(\kappa',T) - \rho_{\rm unctr} = T\bigl[k_1d_1D\rho_{\rm unctr} + (C_\rho-\kappa')T\bigr] > 0.
\end{equation*}
Thus, the estimated basin $\mathcal A_z$ strictly contains the uncontrolled basin of attraction, i.e., $\CA^{\mathrm{unctr}} \subsetneq \mathcal A_z\subseteq\mathcal A_K(t_0)$, indicating the enlargement of the basin under control. Note that the enlargement is guaranteed when stimulation control is initiated at the perturbed state, as in responsive stimulation triggered upon detection~\cite{skarpaas2019brain}. 

\begin{remark}
    As in the nonzero-mean case, the proof of Theorem~\ref{thm:zero-mean-convergence} relies on approximating the unstable periodic solution that characterizes the basin boundary. However, under the zero-mean condition~\eqref{zero-mean-condition}, first-order averaging reproduces the uncontrolled dynamics and thus cannot be directly applied. We therefore use higher-order techniques in the averaging analysis and reconstruct the periodic solution.
\end{remark}

\subsection{Proof of Theorem~\ref{thm:zero-mean-convergence}}

To start with averaging, let $s=(t-t_0)/T$ and define $V(s):=K(t_0+Ts)$.
We introduce the exact periodic transformation
$\rho=e^{-2Tq(s)}R$, where
\begin{equation*}
q'(s)=V(s),
\qquad
\int_0^1q(s)\,ds=0.
\end{equation*}
The transformation yields $R$-dynamics
\begin{equation}\label{R-dynamics}
\frac{dR}{ds} = 2TR\left(
\sigma+2ab e^{-2Tq(s)}R-b e^{-4Tq(s)}R^2
\right).
\end{equation}

Applying the periodic near-identity transformation~\cite[Sec. 2.8]{Sanders2007Averaging}, we have
\begin{equation*}
R=r-8bT^2r^2(a-r)Q(s)+\mathcal O(T^3),
\end{equation*}
where
\begin{equation*}
Q'(s)=q(s),
\qquad
\int_0^1Q(s)\,ds=0.
\end{equation*}

The $R$-dynamics~\eqref{R-dynamics} can be reduced to the truncated third-order averaged system
\begin{equation}\label{higher-order-averaged}
\frac{dr}{ds}= T\overline G_1(r)+T^3\overline G_3(r),
\end{equation}
where
\begin{equation*}
\overline G_1(r)=2r(\sigma+2abr-br^2),
\end{equation*}
and
\begin{equation*}
\overline G_3(r)
=\frac{2}{3}br^2(a-2r)D(4-3D)(k_1d_1)^2.
\end{equation*}

Unlike the first-order averaged equation,
\eqref{higher-order-averaged} contains a correction
that depends on the pulse parameters through
$\overline G_3$. We next show that this averaged
system remains bistable for sufficiently small $T$ and derive its unstable
equilibrium.

\begin{lemma}[Bistability of averaged system]\label{lem:bistability-avg}
Suppose that $-a^2b<\sigma<0$,
with the pulse parameters $k_1,d_1,d_2$ fixed.
There exists $T_b>0$ such that, for every $T\in(0,T_b]$,
the averaged system~\eqref{higher-order-averaged}
is bistable,  and the basin of attraction of $r=0$ is $[0,r_-)$, where the averaged basin boundary $r_-$ satisfies
\begin{equation*}
    r_-=\rho_{\rm unctr}
+\Delta_{\rm avg}T^2+\mathcal O(T^3),
\end{equation*}
where $\lambda:=\sqrt{a^2+\sigma/b}$,
$\rho_{\rm unctr}:=a-\lambda$, and
\begin{equation*}
\Delta_{\rm avg}
=
\frac{D(4-3D)}{6\lambda}
(a-\lambda)(a-2\lambda)(k_1d_1)^2.
\end{equation*}
\end{lemma}

\begin{proof}
Let $\chi:=D(4-3D)(k_1d_1T)^2=\mathcal O(T^2)$.
System~\eqref{higher-order-averaged} becomes
$dr/ds=2TrP(r)$, where
\begin{equation*}
P(r)=\sigma+ab\left(2+\frac{\chi}{3}\right)r
-b\left(1+\frac{2\chi}{3}\right)r^2.
\end{equation*}
For sufficiently small $T$, its discriminant satisfies
$\Delta_P=4b(a^2b+\sigma)+\mathcal O(T^2)>0$.
The roots of $P(r) = 0$ have positive sum and product, so both are positive. Denote these roots by $r_-<r_+$. Since $P(0)<0$ and $P$ is concave, $P<0$ on $(0,r_-)\cup(r_+,\infty)$ and $P>0$ on $(r_-,r_+)$. As in the proof of Lemma~\ref{lem:bistability-pulse}, the system is bistable and the basin of attraction of $r=0$ is $[0,r_-)$.

Finally, substituting
$r_-=\rho_{\rm unctr}+cT^2+\mathcal O(T^3)$
into the equilibrium equation yields
\begin{equation*}
c=-\frac{\overline G_3(\rho_{\rm unctr})}
{\overline G_1'(\rho_{\rm unctr})}
=\Delta_{\rm avg},
\end{equation*}
which completes the proof.
\end{proof}

\begin{lemma}[Existence of periodic solution]\label{lem:zero-mean-periodic-existence}
There exists $T_e'>0$, such that for all $T \in (0,T_e']$, the $R$-dynamics~\eqref{R-dynamics} admit a $1$-periodic solution $R_-^{\rm per}(s;T)$ satisfying
\begin{equation*}
R_-^{\rm per}(0;T)\to\rho_{\rm unctr}\text{ when }T\to0.
\end{equation*}
\end{lemma}

\begin{proof}
Since $\overline G_1(\rho_{\rm unctr}) =0$, $\overline G_1'(\rho_{\rm unctr}) \neq 0$, and $\overline G_1(r)$ is not identically zero, by~\cite[Theorem A]{Novaes_2022}, there exists a $1$-periodic solution $R_-^{\rm per}(s;T)$ of $R$-dynamics~\eqref{R-dynamics}, such that the stated limit above holds.
\end{proof}
We next approximate the corresponding $T$-periodic
solution in the original radial coordinate $\rho$.
Evaluating the near-identity transformation at
$r=r_-$ and $s=0$ gives
\begin{equation*}
R_-^{\rm per}(0;T) = \rho_{\rm unctr} +(\Delta_{\rm avg}+\Gamma_0)T^2 +\mathcal O(T^3),
\end{equation*}
where
\begin{equation*}
\Gamma_0
= -\frac{2}{3}bk_1d_1D(4d_1+2d_2-3)
\rho_{\rm unctr}^2(a-\rho_{\rm unctr}).
\end{equation*}

Next, using the transformation $\rho=e^{-2Tq(s)}R$ with $q(0)=-k_1d_1D/2$, we have
\begin{equation*}
e^{-2Tq(0)}= 1+k_1d_1DT +\frac{1}{2}(k_1d_1D)^2T^2 +\mathcal O(T^3).
\end{equation*}
Since $s=0$ corresponds to $t=t_0$, multiplying the above by $R^{\rm per}_-(0;T)$, we obtain
\begin{equation}\label{zero-mean-periodic-expansion}
\rho_-^{\rm per}(t_0;T) = \rho_{\rm unctr}
+k_1d_1D\rho_{\rm unctr}T
+C_\rho T^2
+\mathcal O(T^3),
\end{equation}
where
\begin{equation*}
C_\rho = \Delta_{\rm avg}+\Gamma_0
+\frac{1}{2}(k_1d_1D)^2\rho_{\rm unctr}.
\end{equation*}
The following lemma shows that the periodic solution
attains its maximum at $t=t_0$.

\begin{lemma}[Upper bound of periodic solution]\label{lem:periodic-max}
Let $\rho_-^{\rm per}(t;T)$ be the unstable
$T$-periodic solution of system~\eqref{radial dynamics}
under the zero-mean gain~\eqref{pulse_zero_mean}.
For sufficiently small $T$,
\begin{equation*}
\max_{t\in[t_0,t_0+T]}\rho_-^{\rm per}(t;T)
=
\rho_-^{\rm per}(t_0;T).
\end{equation*}
\end{lemma}

\begin{proof}
Let $\rho_*(t):=\rho_-^{\rm per}(t;T)$.
Since $\rho_*(t)=\rho_{\rm unctr}+\mathcal O(T)$
uniformly on $[t_0,t_0+T]$ and
$h(\rho_{\rm unctr})=0$, we have
$h(\rho_*(t))=\mathcal O(T)$.
Thus, for sufficiently small $T$,
\begin{equation*}
\begin{aligned}
\dot\rho_*(t)&<0,
&&t\in(t_0,t_0+d_1T),\\
\dot\rho_*(t)&>0,
&&t\in(t_0+d_1T,t_0+DT),
\end{aligned}
\end{equation*}
because $K(t)=k_1>0$ and $K(t)=k_2<0$,
respectively.

If $D<1$, then $K(t)=0$ on $(t_0+DT,t_0+T)$.
Using $q(D)=-k_1d_1D/2$ and
$R_-^{\rm per}(D;T)=\rho_{\rm unctr}+\mathcal O(T^2)$
gives
\begin{equation*}
\rho_*(t_0+DT)
=
\rho_{\rm unctr}
+k_1d_1D\rho_{\rm unctr}T
+\mathcal O(T^2).
\end{equation*}
Hence, for sufficiently small $T$,
$\rho_{\rm unctr}<\rho_*(t_0+DT)<\rho_+$,
where $\rho_+:=a+\sqrt{a^2+\sigma/b}$.
The uncontrolled radial dynamics preserve this interval,
on which $h(\rho)>0$. Therefore,
$\dot\rho_*(t)>0$ on $(t_0+DT,t_0+T)$.

Thus, $\rho_*(t)$ decreases initially and increases
thereafter. Since $\rho_*(t_0+T)=\rho_*(t_0)$, we have 
$\rho_*(t)\leq\rho_*(t_0)$ for $t\in[t_0,t_0+T]$.
\end{proof}

\begin{pfof}{Theorem~\ref{thm:zero-mean-convergence}}
Fix any $\kappa'>0$ and let $\rho_*(t):=\rho_-^{\rm per}(t;T)$. By~\eqref{zero-mean-periodic-expansion},
\begin{equation*}
\rho_*(t_0)-\rho_-(\kappa',T)=\kappa'T^2+\mathcal O(T^3)>0
\end{equation*}
for sufficiently small $T>0$.

Take any solution initialized at time $t_0$ in $\mathcal A_z$ and set $\rho_0=\rho(t_0):=x(t_0)^2+y(t_0)^2$. If $\rho_0=0$, the conclusion is immediate. Otherwise, $\rho_0<\rho_-(\kappa',T)<\rho_*(t_0)$.
Since $\rho_*(t)=\rho_{\rm unctr}+\mathcal O(T)$ uniformly over one period and $0<\rho_{\rm unctr}<a$, periodicity and scalar order preservation give
\begin{equation*}
0<\rho(t)<\rho_*(t)<a, \qquad t\geq t_0,
\end{equation*}
when $T$ is sufficiently small.

\begin{figure*}[t]
\centering
\includegraphics[width=\textwidth]{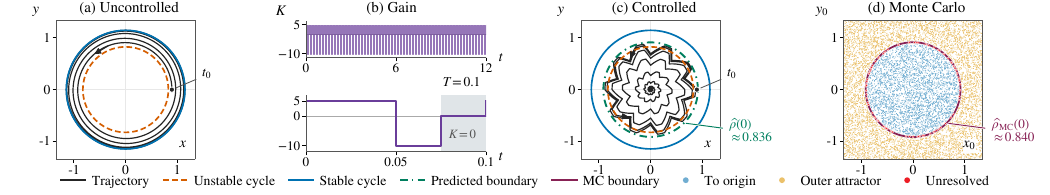}
\caption{Basin enlargement under zero-mean pulsatile feedback.
(a) Uncontrolled trajectory, starting from $\rho_0=0.8$, $\theta_0=0$.
(b) Periodic gain $K(t)$ with a one-period detail.
(c) Controlled trajectory, starting from
$\rho_0=0.8$, $\theta_0=0$.
The green dash-dotted circle in (c) marks the predicted basin boundary.
(d) Monte Carlo results for $10\,000$ initial conditions at $t_0=0$,
with the estimated boundary shown in purple-red.
Blue and orange points indicate convergence to the origin and
approach to the outer attractor, respectively.
Red points (unresolved) highlight samples near the basin boundary.}
\label{fig:zero-mean-basin-enlargement}
\end{figure*}

By the same logarithmic-ratio and one-period-map arguments as in the proof of Theorem~\ref{thm:pulse-convergence}, we obtain \begin{equation*}
\frac{\rho(t)}{\rho_*(t)}\leq\frac{\rho(t_0)}{\rho_*(t_0)}, \qquad t\geq t_0.
\end{equation*}
and $\rho(t_0+nT)\to0$. 
Here, since $h(\rho)\leq\sigma+a^2b$ and $K(t)\geq k_2$, the growth bound becomes $\dot\rho/\rho\leq C:=2(\sigma+a^2b-k_2)>0$ almost everywhere. Thus, for $t\in[t_0+nT,t_0+(n+1)T]$,
\begin{equation*}
\rho(t)\leq e^{CT}\rho(t_0+nT),
\end{equation*}
which yields $\rho(t)\to0$ and hence $(x(t),y(t))\to(0,0)$.

Finally, Lemma~\ref{lem:periodic-max} and periodicity give $\rho_*(t)\leq\rho_*(t_0)$ for all $t\geq t_0$. Combining this with the ratio estimate above yields
\begin{equation*}
\rho(t)\leq\rho(t_0)\frac{\rho_*(t)}{\rho_*(t_0)}\leq\rho(t_0)<\rho_-(\kappa',T), \qquad t\geq t_0.
\end{equation*}
Thus, every solution initialized in $\mathcal A_z$ at stimulation onset remains in $\mathcal A_z$. Choosing a common $T_1>0$ to satisfy all the smallness conditions above completes the proof.
\end{pfof}

\section{Control Design}

This section considers design of pulse parameters for maximizing the basin of attraction.

For the nonzero-mean pulse,
\begin{equation*}
\rho_-^{\rm avg}
=
a-\sqrt{a^2+\frac{\sigma-kd}{b}}
\end{equation*}
increases monotonically with $kd$. Hence, the basin is enlarged by increasing
$k$ or $d$, subject to the bistability condition $kd<a^2b+\sigma$.

For the zero-mean pulse, we have
\begin{equation*}
\rho_-^{\rm per}(t_0;T)-\rho_{\rm unctr}
=
k_1d_1D\,\rho_{\rm unctr}T+\mathcal O(T^2),
\quad D=d_1+d_2,
\end{equation*}
so, for sufficiently small $T$, maximizing the basin reduces to maximizing
$k_1d_1D$.

\begin{theorem}[Optimal zero-mean pulse design]
Let $K_{\max}>0$ denote the maximum pulse amplitude
and $d^+\in(0,1]$ the upper bound on the total duty cycle.
Assume the zero-mean feedback gain~\eqref{pulse_zero_mean},
with $k_1>0$ and $k_2<0$, satisfies
\begin{equation*}
d_1+d_2\leq d^+,\qquad |k_1|,|k_2|\leq K_{\max}.
\end{equation*}
Then, for a fixed $T$, the estimated basin enlargement reaches its maximum if
\begin{equation*}
k_1^*=K_{\max},\qquad
k_2^*=-K_{\max},\qquad
d_1^*=d_2^*=\frac{d^+}{2}.
\end{equation*}
\end{theorem}

\begin{proof}
Let $J:=k_1d_1D$. Maximizing the leading-order basin
enlargement is equivalent to maximizing $J$.
The amplitude bounds and zero-mean condition imply
\begin{equation*}
J=k_1d_1D\leq K_{\max}d_1D,
\qquad
J=(-k_2)d_2D\leq K_{\max}d_2D.
\end{equation*}
Thus,
\begin{equation*}
J\leq K_{\max}D\min\{d_1,d_2\}
\leq \frac{K_{\max}D^2}{2}
\leq \frac{K_{\max}(d^+)^2}{2}.
\end{equation*}
Equality in $J\leq K_{\max}D^2/2$ holds if and only if
$d_1=d_2$ and $k_1=-k_2=K_{\max}$.
Equality in $K_{\max}D^2/2\leq K_{\max}(d^+)^2/2$
holds if and only if $D=d^+$.
Hence, the maximum is achieved when
\begin{equation*}
k_1=K_{\max},\qquad
k_2=-K_{\max},\qquad
d_1=d_2=\frac{d^+}{2},
\end{equation*}
which completes the proof.
\end{proof}

\section{Numerical studies}

We examine basin enlargement under zero-mean pulsatile feedback using
$a=b=1$, $\sigma=-0.9$, and $\omega=2\pi$. The uncontrolled system is
bistable, with uncontrolled basin of attraction of the healthy equilibrium being (see Fig.~\ref{fig:zero-mean-basin-enlargement}(a)) \begin{equation*}
	\CA^{\rm unctr} \coloneq\left\{x,y\in\R\colon x^2 +y^2 < 1 - {\sqrt{10}}/{10}\right\}.
\end{equation*}
We set $T=0.1$ and consider the feedback gain~\eqref{pulse_zero_mean} with $k_1=5,k_2=-10, d_1=0.5$, and $d_2=0.25$,
which satisfies the zero-mean condition (see Fig.~\ref{fig:zero-mean-basin-enlargement}(b)).
Then, it can be calculated from Theorem~\ref{thm:zero-mean-convergence} that the estimated basin becomes 
\begin{equation*}
\mathcal A_z := \left\{x,y\in\mathbb R: x^2+y^2<0.83\right\}.
\end{equation*}
It can be observed from  Fig.~\ref{fig:zero-mean-basin-enlargement}(c) that the controlled basin enlarges the uncontrolled one. Figs.~\ref{fig:zero-mean-basin-enlargement}(a) and~\ref{fig:zero-mean-basin-enlargement}(c) show the uncontrolled and controlled trajectories on the phase plane, respectively, both starting from $\rho_0=0.8$ and $\theta_0=0$. Without control, the trajectory approaches the limit cycle, while under control, it approaches the healthy equilibrium.
Fig.~\ref{fig:zero-mean-basin-enlargement}(d) shows Monte Carlo
results for $10\;000$ initial conditions at $t_0=0$.
Blue and orange points correspond to trajectories approaching
the origin and the outer attractor, respectively, while red points
are concentrated near the basin boundary. It demonstrated that the theoretically-derived basin is very close to the empirical tested basin when $T$ is small.

\section{Conclusion}

This paper shows that both pulsatile feedback control can enlarge the basin of attraction of the stable equilibrium in a bistable neural population model, improving resilience to perturbations that may trigger seizure-like oscillations. We employ high-order averaging techniques for the analysis. Analytical estimate of the basin boundaries enable pulse-parameter selection that maximizes enlargement under amplitude constraints, while numerical simulations validate the theoretical results. These findings could provide a theoretical basis for designing pulsatile stimulation to strengthen the resilience of normal neural activity.
Future work could extend the analysis to strictly zero-mean control inputs and coupled neural populations. Other directions include studying robustness to noise and parameter uncertainty. Experimental validation is also needed to assess the physiological relevance of the theoretical findings.

\bibliographystyle{IEEEtran}
\bibliography{refers}

\end{document}